\documentclass[12pt,reqno]{article}

\usepackage[usenames]{color}
\usepackage{amssymb}
\usepackage{amsmath}
\usepackage{amsthm}
\usepackage{amsfonts}
\usepackage{amscd}
\usepackage{graphicx}

\usepackage[colorlinks=true,
linkcolor=webgreen,
filecolor=webbrown,
citecolor=webgreen]{hyperref}

\definecolor{webgreen}{rgb}{0,.5,0}
\definecolor{webbrown}{rgb}{.6,0,0}

\usepackage{color}
\usepackage{fullpage}
\usepackage{float}

\usepackage{graphics}
\usepackage{latexsym}
\usepackage{epsf}
\usepackage{breakurl}

\DeclareMathOperator{\adh}{Adh}
\def\modd#1 #2{#1\ \mbox{\rm (mod}\ #2\mbox{\rm )}}
\newcommand{\seqnum}[1]{\href{https://oeis.org/#1}{\underline{#1}}}

\begin{document}

\theoremstyle{plain}
\newtheorem{theorem}{Theorem}
\newtheorem{corollary}[theorem]{Corollary}
\newtheorem{lemma}[theorem]{Lemma}
\newtheorem{proposition}[theorem]{Proposition}

\theoremstyle{definition}
\newtheorem{definition}[theorem]{Definition}
\newtheorem{example}[theorem]{Example}
\newtheorem{conjecture}[theorem]{Conjecture}

\theoremstyle{remark}
\newtheorem{remark}[theorem]{Remark}

\title{Words Avoiding Reversed Factors, Revisited}

\author{Lukas Fleischer and Jeffrey Shallit\\
School of Computer Science\\
University of Waterloo \\
Waterloo, ON  N2L 3G1 \\
Canada\\
\href{mailto:shallit@uwaterloo.ca}{\tt shallit@uwaterloo.ca}
}

\maketitle

\begin{abstract}
In 2005, Rampersad and the second author proved a number of theorems about
infinite words $\bf x$ with the property that if $w$ is any
sufficiently long finite factor of $\bf x$, then its reversal $w^R$ is 
{\it not\/}
a factor of $\bf x$.  In this note we revisit these results, reproving
them in more generality, using machine computations only.  Two
different techniques are presented.
\end{abstract}

\section{Introduction}

In this paper we are concerned with certain avoidance properties
of finite and infinite words.

Recall that a word $x$ is said to be a {\it factor} of a word 
$w$ if there exist words $y,z$ such that $w = yxz$.  For example,
the word {\tt act} is a factor of the word
{\tt factor}.  Another term for {\it factor}
is {\it subword}, although this latter term sometimes refers to
a different concept entirely.

For $n \geq 1$
define the property $P_\ell (w)$ of a word $w$ as follows:
\begin{displaymath}
\forall {\text{ factors $x$ of $w$ }}
(|x| \geq \ell)   \implies \text{ ($x^R$ is not a factor
of $w$).}
\end{displaymath}
If $P_\ell(w)$ holds, then we say {\it $w$ avoids reversed
factors of length $\geq \ell$}.    In particular, if $P_\ell (w)$ holds,
then $w$ has no palindromes of length $\geq \ell$.
Clearly $P_1(w)$ holds only for
$w = \varepsilon$, the empty word, so in what follows we always
assume $\ell \geq 2$.  

Define $L_\ell(\Sigma_k)  = \{ w \in \Sigma_k^* \ : \ P_\ell (w) \text{ holds} \} $,
the set of all words over the finite alphabet $\Sigma_k$
avoiding reversed factors of length $\geq \ell$.
In 2005, Rampersad and the second author \cite{Rampersad&Shallit:2005}
proved a number of theorems about $L_\ell(\Sigma_k)$ and related infinite
words.  These results were proved mostly by case-based arguments.
In this paper, we revisit these results, using a new method,
based on finite automata.  Our method is able to prove most of
the results in the previous paper, and more, using a unified
approach.  

A companion paper is \cite{Fleischer&Shallit:2019b}, which
explores the same theme with regard to palindromes.

\section{The language of words avoiding reversed factors is regular}

We define $\Sigma_k = \{ 0, 1, \ldots, k-1 \}$.

The crucial observation is contained in this section.
We show that for every $n \geq 2$ and every $k \geq 1$,
the language $L_\ell(\Sigma_k)$ is regular.

\begin{theorem}
\begin{equation}
\overline{L_\ell(\Sigma_k)}
= \bigcup_{x \in \Sigma_k^\ell} \left( \, \Sigma_k^* \, x \, \Sigma_k^* \cap \Sigma_k^* \, x^R \, \Sigma_k^* \, \right) .
\label{lprime}
\end{equation}
\label{one}
\end{theorem}

\begin{proof}
Suppose $w \not\in L_\ell(\Sigma_k)$.   Then $w$ contains $z$ and $z^R$
as factors for some $z$ with $|z| \geq \ell$.
Writing $z = xy$ with $|x| = \ell$, we
see that $w$ also contains $x$ and $x^R$ as factors, and hence
$w \in \Sigma_k^* \, x \, \Sigma_k^* \cap \Sigma_k^* \, x^R \, \Sigma_k^*$.

On the other hand, suppose $w \in \bigcup_{x \in \Sigma_k^\ell} 
( \Sigma_k^* \, x \, \Sigma_k^* \cap \Sigma_k^* \, x^R \, \Sigma_k^*) $.
Then there exists some $x$ of length $\ell$ such
that $w \in \Sigma_k^* \, x \, \Sigma_k^* \cap \Sigma_k^* \, x^R \, \Sigma_k^*$.
Hence $w$ contains both $x$ and $x^R$ as length-$\ell$ factors,
and so $w  \not\in L_\ell(\Sigma_k)$.  
\end{proof}

\begin{corollary}
The language $L_\ell (\Sigma_k)$ is regular.
\label{two}
\end{corollary}

\begin{proof}
Theorem~\ref{one} shows that $\overline{L_\ell(\Sigma_k)}$ is regular, as
it is the union of regular languages.   So $L_\ell(\Sigma_k)$ is regular.
\end{proof}

Corollary~\ref{two} provides an algorithmic way to characterize all
finite words avoiding reversed factors:  namely, just compute the
minimal DFA $A$ for $L_\ell (\Sigma_k)$.

It also provides a way to characterize the (one-sided) infinite 
words avoiding reversed factors:
since $L_\ell(\Sigma_k)$
is clearly factor-closed (that is, every factor of a word of
$L_\ell(\Sigma_k)$ is also a word of $L_\ell(\Sigma_k)$), it follows that
$A$ has only one non-accepting state, which is necessarily a dead
state.  Without loss of generality, then, we can delete this
dead state, obtaining an automaton $A'$ where every path is labeled
with a word of $L_\ell(\Sigma_k)$ and all words are so represented.
Hence all {\it infinite} words avoiding reversed factors (if any
exist) are given
precisely by the infinite paths through $A'$.  We can characterize
these using the results in Section~\ref{sec3}.

\section{Periodicity}
\label{persec}

Let $\Sigma^\omega$ denote the set of all one-sided infinite words
over the alphabet $\Sigma$.  For a finite nonempty word $x$,
let $x^\omega$ denote the infinite word $xxx\cdots$.  We say
that an element $\bf w$ of $\Sigma^\omega$ is {\it ultimately periodic}
if there exist finite words $y,x$ with $x \not= \varepsilon$ such that
${\bf w} = yx^\omega$.  Otherwise we say $\bf w$ is {\it aperiodic}.

In the expression of an ultimately periodic word in the form
$y x^\omega$, we call $|y|$ the {\it preperiod\/} and $|x|$ the 
{\it period}.

\begin{theorem}
Let $w_0, w_1$ be two noncommuting finite words (that is,
$w_0 w_1 \not= w_1 w_0$).  Define a morphism
$\gamma(i) = w_i$ for $i \in \{ 0, 1 \}$.  
Then ${\bf a} \in \{0,1\}^\omega$ is ultimately periodic
iff $\gamma({\bf a})$ is ultimately periodic.
\label{per}
\end{theorem}

\begin{proof}
Suppose ${\bf a} \in \{0,1\}^\omega$ is ultimately periodic, say
${\bf a} = y z^\omega$.  Then $\gamma({\bf a}) = \gamma(y) \gamma(z)^\omega$,
which shows that $\gamma({\bf a})$ is ultimately periodic with preperiod
$|\gamma(y)|$ and period $|\gamma(z)|$.

For the other direction, let ${\bf a} = a_0 a_1 a_2 \cdots$ and
suppose $\gamma({\bf a}) = {\bf b} = b_0 b_1 b_2 \cdots$ is ultimately periodic,
with preperiod $r$ and period $p$.   Thus $b_i = b_{i+p}$ for all
$i \geq r$.  

Now think of ${\bf b}$ as a concatenation of blocks, each of which
is either $w_0$ or $w_1$.  
Define $d(i) := |\gamma(a_0 a_1 \cdots a_{i-1})|$, and note
that the starting position in $\bf b$ of the $i$'th block,
for $i \geq 0$, is at index $d(i)$.
Let $s$ be the least integer such that $d(s) \geq r$.

By the infinite pigeonhole principle, there must be two integers
$j, k \geq s$, with $j < k$, such that 
\begin{equation}
d(j) \equiv \modd{d(k)} {p}.
\label{mod1}
\end{equation}
The $j$'th block begins at $b_{d(j)}$,
and the $k$'th block begins at $b_{d(k)}$.
The congruence~\eqref{mod1}, together with
the fact that $\bf b$ has period $p$, and
the inequality $d(j), d(k) \geq r$, show that
the two infinite words
$\gamma(a_j a_{j+1} a_{j_2} \cdots) = b_{d(j)} b_{d(j) + 1} b_{d(j) + 2} \cdots$ and
$\gamma(a_k a_{k+1} a_{k+2} \cdots) = b_{d(k)} b_{d(k) + 1} b_{d(k) +2} \cdots $
are identical.

There are now two cases:  either the infinite words $a_j a_{j+1} a_{j+2} \cdots$
and $a_k a_{k+1} a_{k+2} \cdots$ differ, or they are identical.

In the former case, let $i \geq 0$ be the least index such that
$a_{j+i} \not= a_{k+i}$.
Then $a_{j+\ell} = a_{k+\ell}$ for $0 \leq \ell < i$, and so it
follows that $d(j+i) \equiv \modd{d(k+i)} {p}$.  
Thus $b_{d(j+i)} b_{d(j+i)+1} b_{d(j+i)+2} \cdots =
b_{d(k+i)} b_{d(k+i)+1} b_{d(k+i)+2} \cdots $, and so
we have two infinite words,
\begin{displaymath}
{\bf y} = a_{j+i} a_{j+i+1} \cdots \quad \text{ and } \quad
{\bf z} = a_{k+i} a_{k+i+1} \cdots,
\end{displaymath}
one beginning with $0$ and the other beginning with $1$, such that
$\gamma({\bf y}) = \gamma({\bf z})$.  By
\cite[Thm.~2.3.5]{Shallit:2009}, it follows that $w_0$ and $w_1$
commute, a contradiction.

So $a_{j+i} = a_{k+i}$ for all $i \geq 0$, and hence $\bf a$
is ultimately periodic with period $k-j$.
\end{proof}

\section{Adherences}
\label{sec3}

The {\it adherence} $\adh(L)$ of a language is defined as follows:
$$ \adh(L) = \{ {\bf x} \in \Sigma^\omega \ : \ 
\text{every prefix of $\bf x$ is a prefix of some word of $L$} \} .$$
For example, see \cite{Nivat:1978}.

\begin{theorem}
Let $L$ be a regular language.
\begin{itemize}
\item[(a)]  If $L$ is finite then $\adh(L)$ is empty.

\item[(b)] If $L$ is infinite,
but has polynomial growth (that is, there exists a fixed integer $k$
such that the number of length-$n$ words in $L$ is $O(n^k)$),
then $\adh(L)$ is nonempty, but is countable and
contains only ultimately periodic words.

\item[(c)] If $L$ does not have polynomial growth (informally,
$L$ has exponential growth), then $\adh(L)$ is uncountable and
contains uncountably many aperiodic words.
\end{itemize}
\label{adh-thm}
\end{theorem}

\begin{proof}
\leavevmode
\begin{itemize}
\item[(a)]  Trivial.

\item[(b)]  By combining \cite[Prop.~3]{Lecomte&Rigo:2002} with
\cite[Lemma 2.2]{Bell&Hare&Shallit:2018}, we see that $\adh(L)$ is
countable iff $L$ has polynomial growth.  Furthermore, the proof
of \cite[Prop.~3]{Lecomte&Rigo:2002} (specifically, the displayed line
following Eq.~(6) on p.~20 of that paper) actually shows that
$\adh(L)$ consists only of ultimately periodic words.

\item[(c)]  By combining \cite[Prop.~3]{Lecomte&Rigo:2002} with
\cite[Lemma 2.3]{Bell&Hare&Shallit:2018}, we see that
$\adh(L)$ is uncountable iff $L$ has exponential growth.
Since there are only a countable number of ultimately periodic
words, it follows that $\adh(L)$ contains uncountably many
aperiodic words.

\end{itemize}
\end{proof}

\section{Applications}

Let us now turn to reproving the principal theorems from
\cite{Rampersad&Shallit:2005}.  For many of these theorems, we can
employ the following strategy:  
use {\tt Grail}, a software package for manipulating automata 
\cite{Raymond&Wood:1994},
to construct a DFA $M$ corresponding to the regular expression
in Eq.~\eqref{lprime}, and from this obtain a DFA $M'$
for $L_\ell(\Sigma_k)$.  The
infinite words avoiding reversed
factors of length $\geq n$ are then given by all infinite paths
through the digraph of the transition diagram of $M'$.   
Using Theorem~\ref{adh-thm}, we can characterize the infinite words.
Using depth-first search, the finiteness of $L(M')$ can be determined
trivially.  The distinction between polynomial and exponential
growth can be determined efficiently using the methods
detailed in \cite{Gawrychoski&Krieger&Rampersad&Shallit:2010}:
call a state $q$ {\it birecurrent} if there are at least two distinct 
noncommuting words, $x_0$ and $x_1$, taking state $q$ to $q$.   
By Theorem~\ref{per}, if there is
a birecurrent state, we can find an explicit example of an aperiodic
infinite word labeled by an infinite
path through the automaton by replacing the $0$'s (resp., the $1$'s)
in any aperiodic binary word with $x_0$ (resp., $x_1$).
For example,
we can take ${\bf w} = {\bf t} = 01101001 \cdots$, the Thue-Morse
word \cite{Thue:1912,Berstel:1995}.

On the other hand,
if $L(M')$ has polynomial growth, then there are no birecurrent states.
In this case, only periodic infinite words with the
given avoidance properties exist.

In practice, creating the DFA from the regular expression in Eq.~\eqref{lprime}
is not completely straightforward, however,
as exponential blowup is observed
in some formulations.  By experimenting, we found that the
following technique works:  using de Morgan's law,
we rewrite Eq.~\eqref{lprime} as
$$ 
L_\ell(\Sigma)
= \bigcap_{x \in \Sigma_k^n} \left( \, \overline{\Sigma_k^* \, x \, \Sigma_k^*} \cup 
\overline{\Sigma_k^* \, x^R \, \Sigma_k^*} \, \right) , 
$$
and construct minimal DFA's for each individual term of the intersection.
Clearly it suffices to perform the intersection only for those
$x$ for which $x$ is lexicographically equal to or smaller than
$x^R$.
We then iteratively intersect the resulting DFA's term-by-term.
Although intermediate results can be quite large (thousands
of states), the final
DFA so produced is relatively small.

We used a short program written in 
Dyalog APL to create a Linux shell script with
the individual {\tt Grail} commands.  
We used {\tt Grail}, version 3.3.4 \cite{Campeanu:2019}.  
Running this script
creates a text file describing a DFA for $L_\ell (\Sigma_k)$.
We identify the unique nonaccepting
state in the result, and delete lines referencing this state from the
text file.  We then used another
Dyalog APL program to convert this text file to a file in
GraphViz format that can be used to display the automaton.

Since we explicitly construct the DFA for $L_\ell (\Sigma_k)$, another
benefit to our approach is as follows.  Using standard
techniques (e.g., \cite[\S 3.8]{Shallit:2009}), we can enumerate
the number of words of length $n$ in the language.   We briefly
sketch how this can be done.

Once the automaton $A = (Q, \Sigma, \delta, q_0, F)$ for $L_\ell(\Sigma_k)$
is known, we can create a useful 
matrix $r \times r$ matrix $M$ from it as follows (where
$Q = \{q_0, \ldots, q_{r-1} \}$ and $r = |Q|$):
$$M[i,j] = | \{ a \in \Sigma_k \ : \ \delta(q_i, a) = q_j \} |.$$
This matrix $M$ has the property that $M^n[i,j]$ is the number of
words taking $A$ from state $q_i$ to state $q_j$.  
The minimal polynomial of $M$ then gives a recurrence for
the number of length-$n$ words that $A$ accepts.  For the
details, see \cite{Fleischer&Shallit:2019b}.  Thus our method
allows an automated way to obtain the number of length-$n$
words in $L_\ell(\Sigma_k)$ and its asymptotic growth rate.

We now reprove the theorems from \cite{Rampersad&Shallit:2005}.

\subsection{Alphabet size 3}

\begin{theorem}
\label{tern_per}
There exists an infinite word $\mathbf{w}$ over $\Sigma_3$ such that if
$x$ is a factor of $\mathbf{w}$ and $|x| \geq 2$, then $x^R$ is not a
factor of $\mathbf{w}$.  Furthermore, $\mathbf{w}$ is unique up to
permutation of the alphabet symbols.
\end{theorem}

\begin{proof}
We use the following Linux shell script to create the automaton:
{\footnotesize
\begin{verbatim}
# making automaton for K = 3; N = 2
echo "00"
echo "(0+1+2)*00(0+1+2)*" | ./retofm | ./fmdeterm | ./fmmin | ./fmcment > d0
./fmstats d0
echo "01"
echo "(0+1+2)*01(0+1+2)*" | ./retofm | ./fmdeterm | ./fmmin | ./fmcment > a1
echo "(0+1+2)*10(0+1+2)*" | ./retofm | ./fmdeterm | ./fmmin | ./fmcment > b1
./fmunion a1 b1 | ./fmdeterm | ./fmmin > c1
./fmcross d0 c1 | ./fmdeterm | ./fmmin > d1
./fmstats d1
echo "02"
echo "(0+1+2)*02(0+1+2)*" | ./retofm | ./fmdeterm | ./fmmin | ./fmcment > a2
echo "(0+1+2)*20(0+1+2)*" | ./retofm | ./fmdeterm | ./fmmin | ./fmcment > b2
./fmunion a2 b2 | ./fmdeterm | ./fmmin > c2
./fmcross d1 c2 | ./fmdeterm | ./fmmin > d2
./fmstats d2
echo "11"
echo "(0+1+2)*11(0+1+2)*" | ./retofm | ./fmdeterm | ./fmmin | ./fmcment > c3
./fmcross d2 c3 | ./fmdeterm | ./fmmin > d3
./fmstats d3
echo "12"
echo "(0+1+2)*12(0+1+2)*" | ./retofm | ./fmdeterm | ./fmmin | ./fmcment > a4
echo "(0+1+2)*21(0+1+2)*" | ./retofm | ./fmdeterm | ./fmmin | ./fmcment > b4
./fmunion a4 b4 | ./fmdeterm | ./fmmin > c4
./fmcross d3 c4 | ./fmdeterm | ./fmmin > d4
./fmstats d4
echo "22"
echo "(0+1+2)*22(0+1+2)*" | ./retofm | ./fmdeterm | ./fmmin | ./fmcment > c5
./fmcross d4 c5 | ./fmdeterm | ./fmmin > d5
./fmstats d5
cp d5 aut32.txt
\end{verbatim}
}
\noindent which, after deleting lines corresponding to the dead state numbered
4, gives the following {\tt Grail} output:
\begin{verbatim}
(START) |- 0 
0 0 1        
0 1 2        
0 2 3        
1 1 5        
1 2 6        
2 0 7        
2 2 8        
3 0 9        
3 1 10       
5 2 8        
6 1 10       
7 2 6        
8 0 9        
9 1 5        
10 0 7       
0 -| (FINAL) 
1 -| (FINAL) 
2 -| (FINAL) 
3 -| (FINAL) 
5 -| (FINAL) 
6 -| (FINAL) 
7 -| (FINAL) 
8 -| (FINAL) 
9 -| (FINAL) 
10 -| (FINAL)
\end{verbatim}
which is depicted below.
\begin{center}
\begin{figure}[H]
\includegraphics[width=5in]{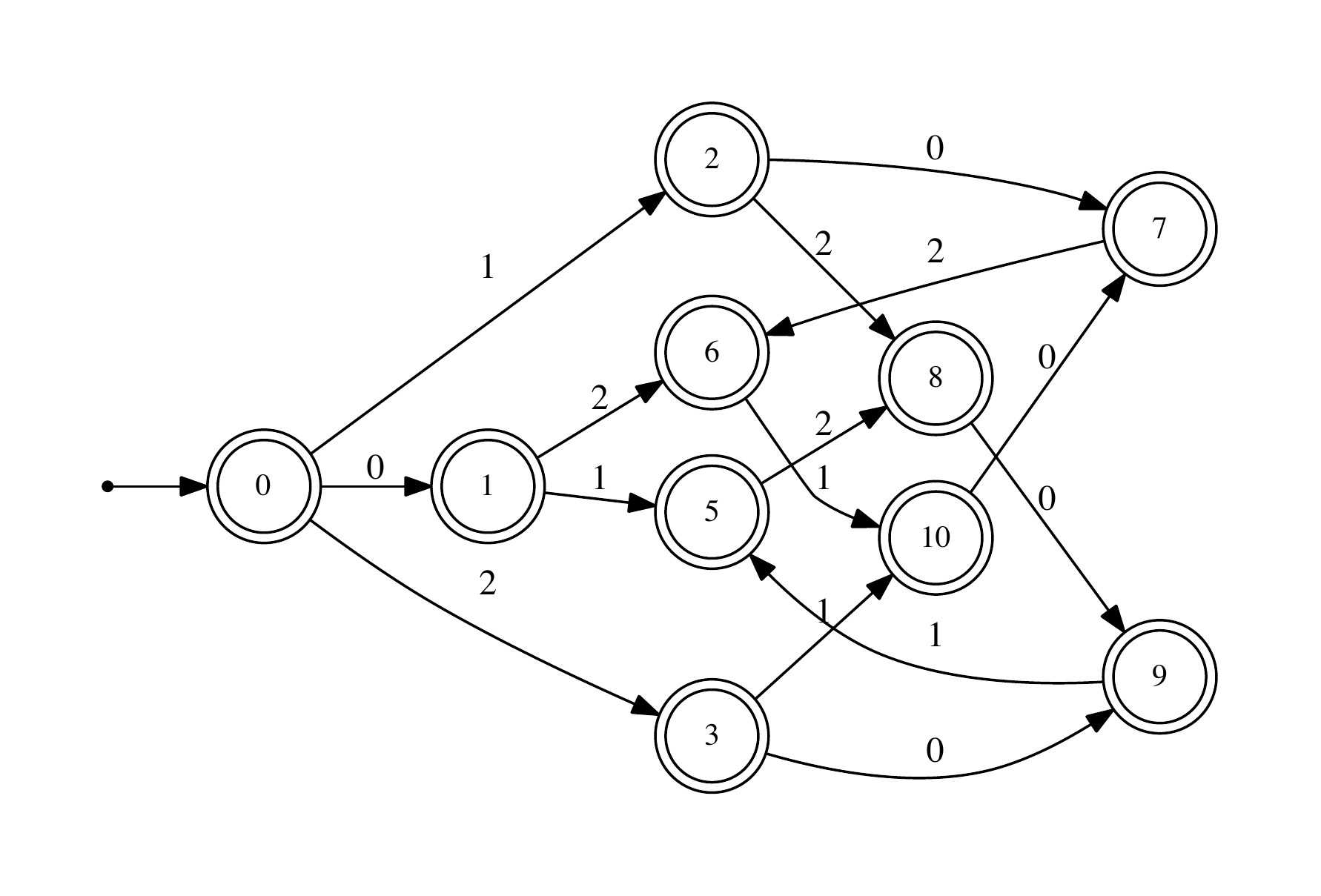}
\caption{Automaton for $L_2(\Sigma_3)$.  Dead state, numbered 4, omitted.}
\end{figure}
\end{center}
As the reader can now easily verify, the set of finite words accepted
are the prefixes of
$$(012)^* + (021)^* + (102)^* + (120)^* + (201)^* + (210)^* .$$
The corresponding set of infinite words is then
$$(012)^\omega + (021)^\omega + (102)^\omega + (120)^\omega + (201)^\omega + (210)^\omega .$$
\end{proof}

In subsequent theorems, we omit providing the shell scripts and outputs
from {\tt Grail}, but the reader can obtain them from
the second author's web page \\
\centerline{\url{https://cs.uwaterloo.ca/~shallit/papers.html} \ .}

\begin{theorem}
There exists an aperiodic infinite word $\mathbf{w}$ over
$\Sigma_3$ such that if $x$ is a factor of $\mathbf{w}$ and $|x| \geq 3$,
then $x^R$ is not a factor of $\mathbf{w}$.
\label{thm2}
\end{theorem}

\begin{proof}
As above, we create the DFA for $L_3(\Sigma_3)$.
Although the intermediate automata have as many as 1033 states,
the final automaton has only 20 states (including the dead state).
It is depicted below.
\begin{center}
\begin{figure}[H]
\includegraphics[width=6in]{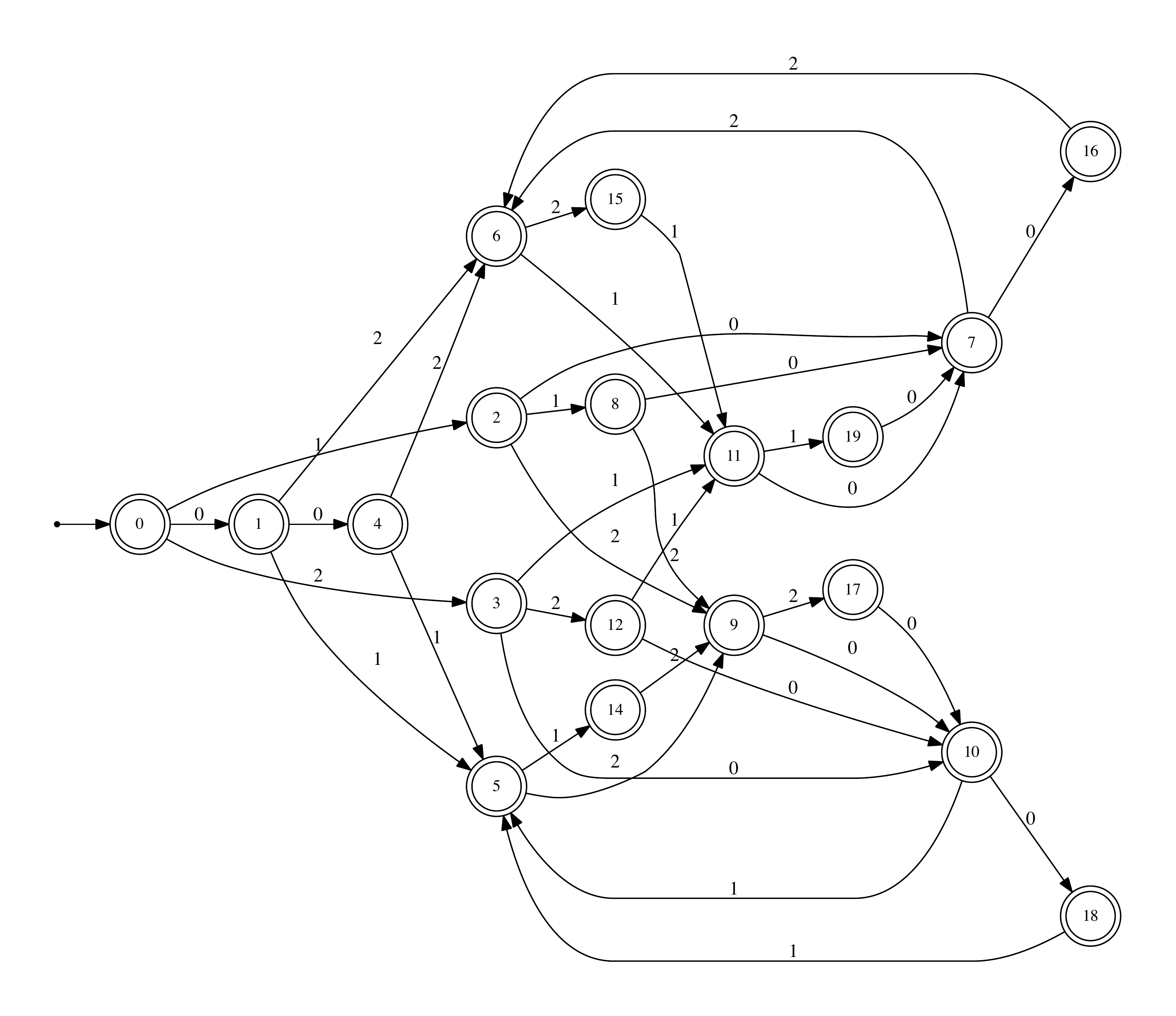}
\caption{Automaton for $L_3(\Sigma_3)$.  Dead state numbered 13, omitted.}
\end{figure}
\end{center}
Then, for example, state $9$ is a birecurrent state, with the
corresponding cycles labeled by $x_0 = 0012$ and $x_1 = 0112$.
It follows that every word in $\{ 0012, 0112 \}^\omega$ avoids
reversed factors of length $\ell \geq 3$, and uncountably many of these
are aperiodic.
\end{proof}

Let the Fibonacci numbers be defined, as usual, 
by the recurrence $F_n = F_{n-1} + F_{n-2}$, together
with the initial conditions $F_0 = 0$ and $F_1 = 1$.

\begin{theorem}
The number $r_{33}(n)$ of length-$n$
words in $L_3(\Sigma_3)$ 
is $6F_{n+1}$ for $n \geq 3$.
\label{m33}
\end{theorem}

\begin{proof}
We create the $20 \times 20$ matrix $M$
corresponding to the transitions of
$L_3 (\Sigma_3)$.  Its minimal polynomial is
$p(X) = X^3 (X-3)(X^2 - X - 1)(X^4 + X^3 + 2X^2 + 2X+ 1)$.
It now follows that $r_{33} (n)$ can be expressed as
a linear combination of $n$'th powers of the zeros of
$p(X)$.  We can determine the coefficients of this
linear combination by solving a linear system, using
the computed values of the first 10 terms of $r_{33} (n)$.
From this, the result easily follows.
\end{proof}

\subsection{Alphabet size 2}

\begin{theorem}
Let $n \leq 4$ and let $w$ be a word over $\Sigma_2$ such that
if $x$ is a factor of $w$ and $|x| \geq n$, then $x^R$ is not
a factor of $w$.  Then $|w| \leq 8$.
\end{theorem}

\begin{proof}
As above, we create the DFA for $L_4(\Sigma_2)$.
It is depicted below, and we easily see that the longest words
accepted are of length $8$.
\begin{center}
\begin{figure}[H]
\includegraphics[width=6in]{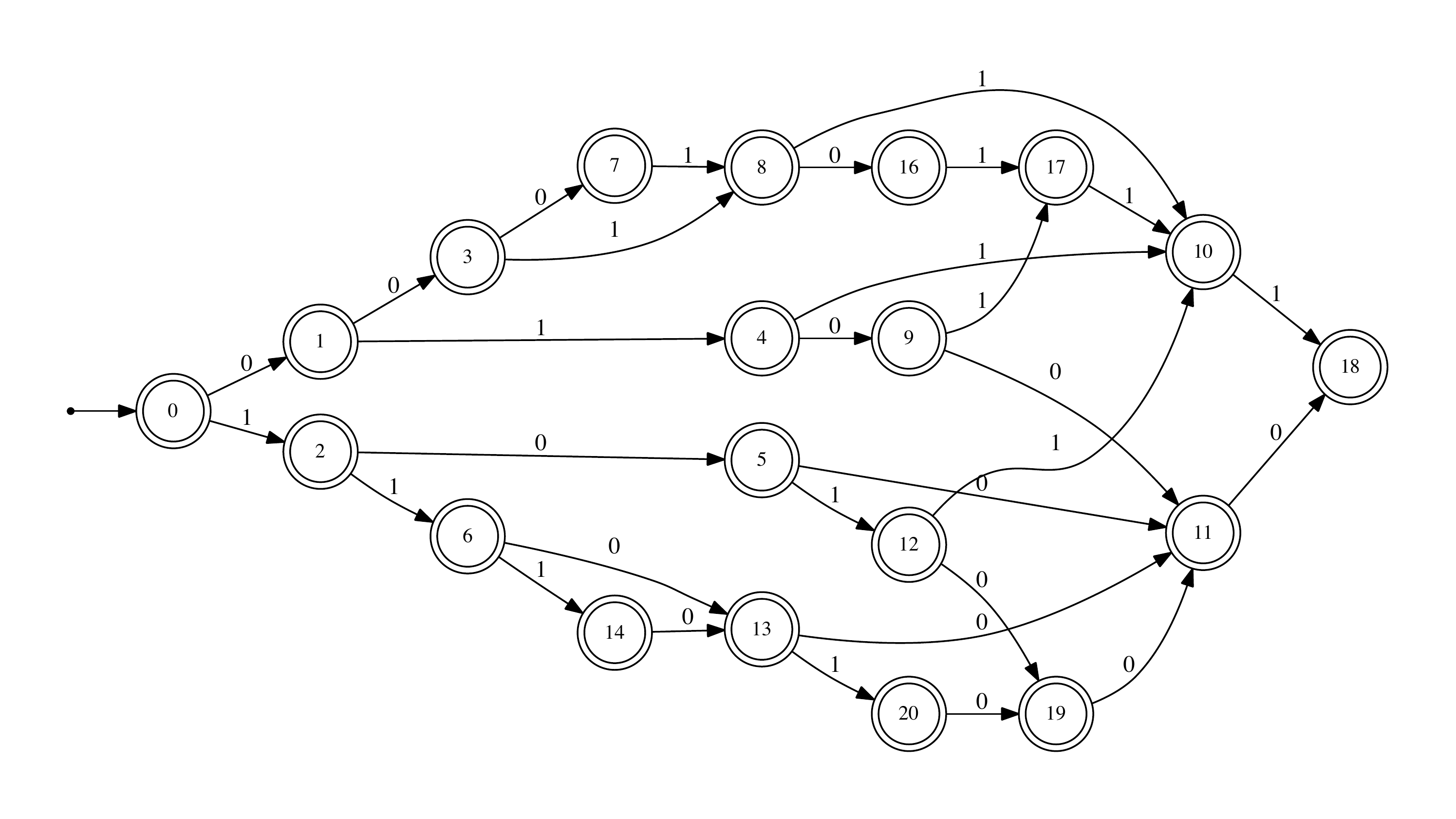}
\caption{Automaton for $L_4(\Sigma_2)$.  Dead state, numbered 15, omitted.}
\end{figure}
\end{center}
\end{proof}

\begin{theorem}
\label{geq5}
There exists an infinite word $\mathbf{w}$ over $\Sigma_2$ such that if
$x$ is a factor of $\mathbf{w}$ and $|x| \geq 5$, then $x^R$ is not a
factor of $\mathbf{w}$.  
\end{theorem}

\begin{proof}
As above, we create the DFA for $L_5(\Sigma_2)$.
Although the intermediate automata produced have as many as 598
states, the final DFA has only 59 states (including the dead state).
\begin{center}
\begin{figure}[H]
\includegraphics[width=6in]{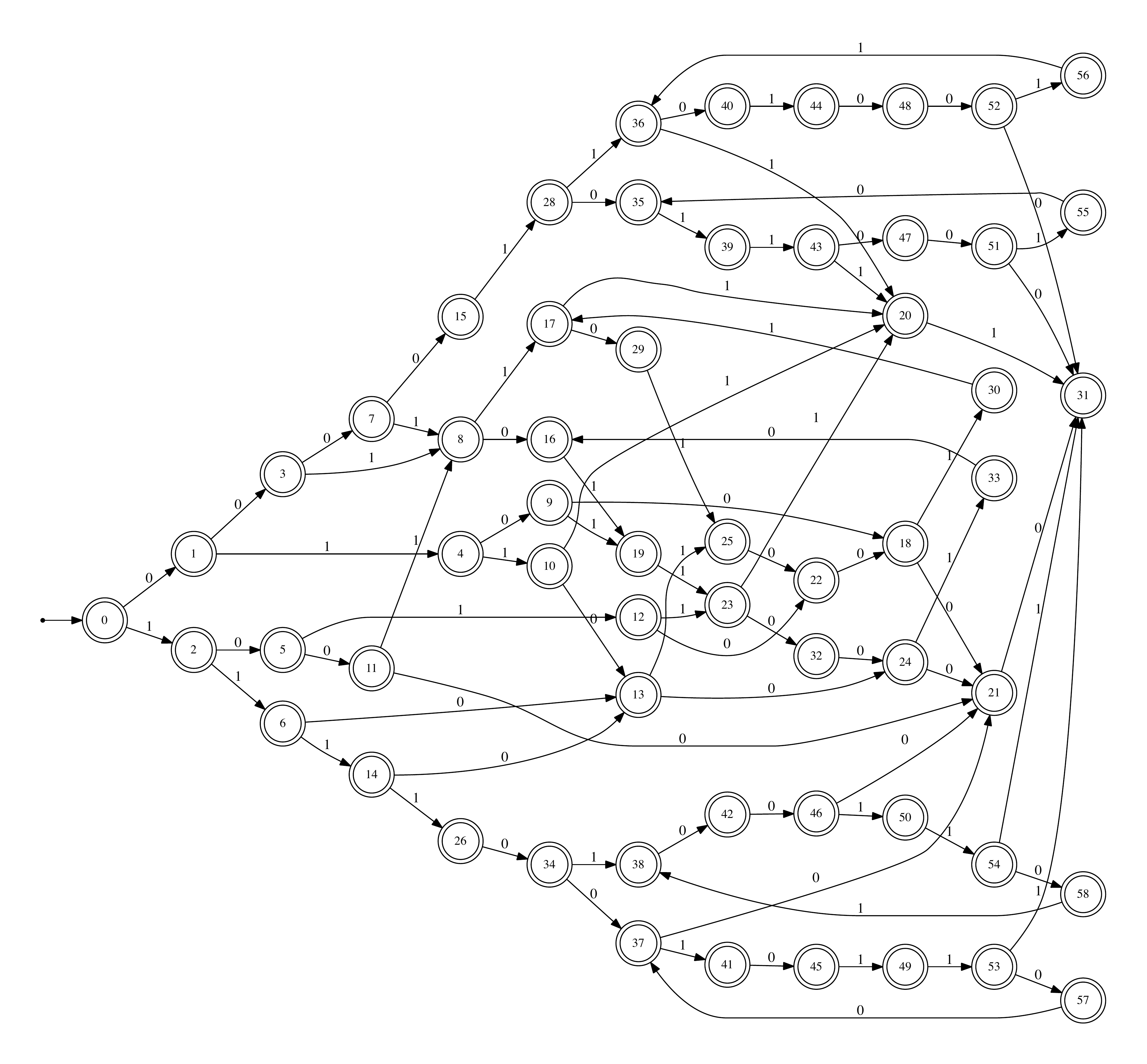}
\caption{Automaton for $L_5(\Sigma_2)$.  Dead state, numbered 27, omitted.}
\end{figure}
\end{center}
Then, for example, $000011(010011)^\omega$ is an infinite path in this
DFA.
\end{proof}

\begin{remark}
We can see by inspection that there are no birecurrent states in this
automaton.  Hence all infinite words satisfying the property of
Theorem~\ref{geq5} are periodic.   
\end{remark}

\begin{theorem}
The number $r_{25} (m)$ of length-$m$ words in $L_5 (\Sigma_2)$ 
is given by
$$
r_{25} (m) = \begin{cases}
	30, & \text{if $m \equiv \modd{0} {6}$}; \\
	32, & \text{if $m \equiv \modd{1,2,3} {6}$ and $m \geq 7$}; \\
	34, & \text{if $m \equiv \modd{4} {6}$ and $m \geq 10$}; \\
	36, & \text{if $m \equiv \modd{5} {6}$ and $m \geq 11$} .
	\end{cases}
$$
\end{theorem}

\begin{proof}
As in the proof of Theorem~\ref{m33}, we can build the $59 \times 59$
matrix corresponding to the automaton, and determine its minimal
polynomial $p(X) = X^6 (X^6-1)(X-2)$.  As before we can express
$r_{25} (m)$ as a linear combination of the $m$'th powers of the
zeros of $p$.  The result now easily follows.
\end{proof}

\begin{theorem}
There exists an aperiodic infinite word $\mathbf{w}$ over
$\Sigma_2$ such that if $x$ is a factor of $\mathbf{w}$ and $|x| \geq 6$,
then $x^R$ is not a factor of $\mathbf{w}$.
\label{thm26}
\end{theorem}

Here our previous approach does not succeed in a reasonable length
of time, because the intermediate automata grow too large 
(at least hundreds of thousands of states).
We describe an alternative approach that produces the desired DFA
for $L_6(\Sigma_2)$.  

We can construct a DFA for $L_\ell(\Sigma_k)$ directly as follows:  
it suffices to record, in the state, the subset of length-$n$
factors seen so far, and the last $n-1$ symbols seen (or shorter
prefix, if $n-1$ symbols have not yet been seen).   Upon reading
a new symbol, the DFA updates the subset of factors and the last
$n-1$ symbols seen.  So the total number of states
is $2^{k^n} \cdot (1+k +k^2 + \cdots + k^{n-1})$.   
The final states correspond to those
subsets not containing both a word and its reversal.  

For our particular case of $k = 2$, $n = 6$, this gives a
DFA with $63 \cdot 2^{64}$ states, which is evidently too large to
manipulate effectively.   However, many of these states will
be unreachable from the
start state.  Instead, we can construct the reachable states
in a breadth-first manner, using a queue.   We wrote a Dyalog APL
program to construct the automaton; it has 63705 states (not 
including the dead state).   We then minimized this automaton
using {\tt Grail}, and we obtained an automaton $A$ with 7761 states
(not including the dead state).  This automaton is much too big to
display here, but can be obtained from the website of the second
author.

State 980 is a birecurrent state, with the corresponding cycles
labeled by $0001011$ and $1001011$.    Now we can complete
the proof of Theorem~\ref{thm26}.

\begin{proof}
As before, we can produce an explicit example of an aperiodic infinite
word satisfying the given conditions by applying the morphism
$0 \rightarrow 0001011$, $1 \rightarrow 1001011$ to any aperiodic
binary word, such as the Thue-Morse word.
\end{proof}

As suggested by the size of the minimal automaton $A$,
it turns out that the structure of the language
$L_6(\Sigma_2)$ is very complicated.    A natural problem is to
give a recurrence enumerating the number $r_{26}(n)$ of length-$n$ words
in $L_6(\Sigma_2)$.   Even this is not so easy; it turns out that
$r_{26}(n)$ satisfies a linear recurrence of order 195.

We describe how this can be proved.  The first step is to compute
the minimal polynomial of the matrix $M$ corresponding to $A$.
We were not able to compute this with {\tt Maple} 2017 (X86 64 LINUX),
so we turned to the software {\tt LinBox} \cite{Dumas:2019}.  It computed
the minimal polynomial as the following polynomial of degree 239:
\begin{align*}
& X^{18} (X - 2) (X - 1) (X + 1) (X^2 + 1) (X^4 + 1) (X^2 - X + 1) (X^2 + X + 1)(X^4 - X^2 + 1) \times \\
&(X^6 + X^3 + 1) (X^8 - X^2 - 1) (X^8 + X^2 - 1) (X^9 - X^2 - 1) (X^{10} - X^2 - 1) (X^{12} - X^2 - 1) \times \\
&(X^{12} - X^3 - 1) (X^{12} - X^4 - 1) (X^{12} - X^5 - 1) (X^{12} - X^6 - 1) (X^4 - X^3 + X^2 - X + 1) \times \\
&(X^4 - X^3 + X^2 - X + 1) (X^4 + X^3 + X^2 + X + 1) (X^7 - X^6 + X^4 - X^3 - 1)\times \\
& (X^{10} - X^3 - X^2 - X - 1) 
(X^{10} - X^8 + X^6 - X^4 - 1) (X^{16} - X^9 - X^7 - X^4 + 1) \times \\
& (X^{16} - X^{10} - X^6 - X^4 + 1) 
(X^{10} - X^4 - 2 X^3 - 2 X^2 - 2 X - 1)
(X^{10} - X^8 + X^6 - 2 X^4 + X^2 - 1) \times \\
&(X^6 + X^5 + X^4 + X^3 + X^2 + X + 1) (X^{10} - X^8 + X^6 - X^4 - X^3 + X^2 - 1).
\end{align*} 
From this one can compute a linear recurrence of order $239$ for the
sequence $r_{26}(n)$.  However, using the techniques from
\cite{Fleischer&Shallit:2019b}, we can find the optimal linear
recurrence, which arises from the following
degree-$195$
divisor of the
minimal polynomial:
\begin{align*}
& (X - 1) (X^2 + 1) (X^2 - X + 1) (X^2 + X + 1) (X^4 - X^2 + 1) (X^8 - X^2 - 1) (X^8 + X^2 - 1) \times \\
& (X^9 - X^2 - 1) (X^{10} - X^2 - 1) (X^{12} - X^2 - 1) (X^{12} - X^3 - 1) (X^{12} - X^4 - 1) (X^{12} - X^5 - 1) \times \\
& (X^{12} - X^6 - 1) (X^7 - X^6 + X^4 - X^3 - 1) (X^{10} - X^3 - X^2 - X - 1) (X^{10} - X^8 + X^6 - X^4 - 1) \times \\
& (X^{16} - X^9 - X^7 - X^4 + 1) (X^{16} - X^{10} - X^6 - X^4 + 1) (X^{10} - X^4 - 2 X^3 - 2 X^2 - 2 X - 1)  \times \\
& (X^{10} - X^8 + X^6 - 2 X^4 + X^2 - 1) (X^{10} - X^8 + X^6 - X^4 - X^3 + X^2 - 1) .
\end{align*}
The largest real zero of this polynomial is
$\alpha \doteq 1.305429354041958520199761719029$,
where $\alpha$ is the positive real zero of 
$X^{10} - X^4 - 2 X^3 - 2 X^2 - 2 X - 1$.
It follows that $r_{26}(n) \sim c \alpha^n$, where $c \doteq 15.0313407$.

\begin{remark} 
The sequence $r_{26} (n)$ is sequence 
\seqnum{A330012} in the {\it On-Line Encyclopedia of Integer
Sequences} (OEIS) \cite{Sloane:2019}.
\end{remark}

\subsection{Alphabet size 4}

Inexplicably, the paper \cite{Rampersad&Shallit:2005} did not handle
the case of alphabet size $4$ (or more precisely, it only
considered the case of squarefree words).  We consider the alphabet
size $4$ case now.

\begin{theorem}
There are uncountably many infinite words over $\Sigma_4$ avoiding reversed
factors for length $\ell \geq 2$.
\end{theorem}

\begin{proof}
We construct the automaton as in Theorem~\ref{thm26}.
The resulting automaton has
449 states and is minimal.  State 360 is birecurrent, with
paths $x_0 = 0123$ and $x_1 = 0120123$.
\end{proof}

\begin{corollary}
Let $r_{42} (n)$ denote the number of length-$n$
words over $\Sigma_4$ avoiding reversed
factors of length $\ell \geq 2$.  Then
\begin{align*}
 & (r_{42}(0), r_{42}(1), \ldots, r_{42}(16)) = \\
 & \quad\quad\quad (1,4,12,24,48,96,168,264,456,720,1056,1656,2520,3600,5352,7944,11256)
 \end{align*}
and
\begin{align*}
r_{42}(n) &= r_{42}(n-1) + 5r_{42}(n-3) - 3r_{42}(n-4) - 2r_{42}(n-5) - 8r_{42}(n-6) + r_{42}(n-7) +  \\
& \quad 6r_{42}(n-8) + 5r_{42}(n-9) + 2r_{42}(n-10) - 4r_{42}(n-11) - 2r_{42}(n-12) 
\end{align*}
for $n \geq 17$.
Asymptotically we have $r_{42}(n) = C \cdot \alpha^n$, where
$\alpha \doteq 1.395336944$ is the largest
real zero of $X^4 -2X - 1$ and $C \doteq 71.2145756$.
\end{corollary}

\begin{proof}
We computed the minimal polynomial of the associated matrix as
above, using {\tt Maple}.  It is
$$X^5(X-1)(X-4)(X+1)(X^2+1)(X^3-2)(X^4-2X-1)(X^2+X+1)(X^4-X-1).$$
Using a technique discussed in
\cite{Fleischer&Shallit:2019b}, we can find the annihilator 
for the sequence, which is
$$ (X-1)(X^3 - 2)(X^4 - 2X - 1)(X^4 - X - 1).$$
Expanding the coefficients of this polynomial gives us
the recurrence.  The largest real root is that of $X^4 - 2X-1$.
\end{proof}

\begin{remark}
The sequence $r_{42} (n)$ is sequence \seqnum{A330011} in the
OEIS.
\end{remark}

\section{Code}

All of the shell scripts, {\tt Maple} code, {\tt LinBox} code,
and automata discussed in the paper
are available at the website of the second author,\\
\centerline{\url{https://cs.uwaterloo.ca/~shallit/papers.html} \ . } 

\newcommand{\noopsort}[1]{} \newcommand{\singleletter}[1]{#1}


\begin{thebibliography}{10}

\bibitem{Bell&Hare&Shallit:2018}
J.~Bell, K.~Hare, and J.~Shallit.
\newblock When is an automatic set an additive basis?
\newblock {\em Proc. Amer. Math. Soc. Ser. B} {\bf 5} (2018), 50--63.

\bibitem{Berstel:1995}
J.~Berstel.
\newblock {\em Axel {Thue's} Papers on Repetitions in Words: a Translation}.
\newblock Number~20 in Publications du Laboratoire de Combinatoire et
  d'Informatique {Math\'ematique}. Universit\'e du Qu\'ebec \`a Montr\'eal,
  February 1995.

\bibitem{Campeanu:2019}
C.~C\^{a}mpeanu et~al.
\newblock {\tt Grail} software package.
\newblock Available from \url{http://www.csit.upei.ca/theory/}, 2019.

\bibitem{Dumas:2019}
J.-G. Dumas, C.~Pernet, et~al.
\newblock Project {{\tt LinBox}}.
\newblock Software available at \url{https://linalg.org/index.html}, 2019.

\bibitem{Fleischer&Shallit:2019b}
L.~Fleischer and J.~Shallit.
\newblock Words with few palindromes, revisited.
\newblock Preprint, 2019.

\bibitem{Gawrychoski&Krieger&Rampersad&Shallit:2010}
P.~Gawrychoski, D.~Krieger, N.~Rampersad, and J.~Shallit.
\newblock Finding the growth rate of a regular or context-free language in
  polynomial time.
\newblock {\em Internat. J. Found. Comp. Sci.} {\bf 21} (2010), 597--618.

\bibitem{Lecomte&Rigo:2002}
P.~Lecomte and M.~Rigo.
\newblock On the representation of real numbers using regular languages.
\newblock {\em Theory Comput. Systems} {\bf 35} (2002), 13--38.

\bibitem{Nivat:1978}
M.~Nivat.
\newblock Sur les ensembles de mots infinis engendr{\'e}s par une grammaire
  alg{\'e}brique.
\newblock {\em RAIRO Inform. Th\'eor. App.} {\bf 12} (1978), 259--278.

\bibitem{Rampersad&Shallit:2005}
N.~Rampersad and J.~Shallit.
\newblock Words avoiding reversed subwords.
\newblock {\em J. Combin. Math. Combin. Comput.} {\bf 54} (2005), 157--164.

\bibitem{Raymond&Wood:1994}
D.~Raymond and D.~Wood.
\newblock {\it Grail}: A {C++} library for automata and expressions.
\newblock {\em J. Symbolic Comput.} {\bf 17} (1994), 341--350.

\bibitem{Shallit:2009}
J.~Shallit.
\newblock {\em A Second Course in Formal Languages and Automata Theory}.
\newblock Cambridge University Press, 2009.

\bibitem{Sloane:2019}
N.~J.~A. Sloane et~al.
\newblock {The On-Line Encyclopedia of Integer Sequences}.
\newblock Electronic resource available at \url{https://oeis.org}, 2019.

\bibitem{Thue:1912}
A.~Thue.
\newblock {\"Uber} die gegenseitige {Lage} gleicher {Teile} gewisser
  {Zeichenreihen}.
\newblock {\em Norske vid. Selsk. Skr. Mat. Nat. Kl.} {\bf 1} (1912), 1--67.
\newblock Reprinted in {\it Selected Mathematical Papers of Axel Thue}, T.
  Nagell, editor, Universitetsforlaget, Oslo, 1977, pp.~413--478.

\end{thebibliography}
\end{document}